\documentclass[12pt,letterpaper]{article}

\usepackage{amsmath,amsfonts,amsthm,amssymb,stmaryrd,relsize}
\usepackage{multirow,blkarray, bigstrut}  

\theoremstyle{plain}

\numberwithin{equation}{section}
\newtheorem{thm}{Theorem}[section]
\newtheorem{lem}[thm]{Lemma}
\newtheorem{cor}[thm]{Corollary}

\newenvironment{exam}[1]
{\begin{flushleft}\textbf{Example #1}.\enspace}%
{\end{flushleft}}

\allowdisplaybreaks  

\newcommand{\complex}{{\mathbb C}}
\newcommand{\Natural}{{\mathbb N}}
\newcommand{\real}{{\mathbb R}}

\newcommand{\minusone}{{\mathbin{\text{--}1}}} 
\newcommand{\binle}{{\mathbin{\le}}} 
\newcommand{\binperp}{{\mathbin{\perp}}} 

\newcommand{\ahat}{\widehat{A}}

\newcommand{\rmtr}{\mathrm{tr\,}}
\newcommand{\rmdiag}{\mathrm{diag\,}}

\newcommand{\escript}{\mathcal{E}}
\newcommand{\fscript}{\mathcal{F}}
\newcommand{\lscript}{\mathcal{L}}
\newcommand{\pscript}{\mathcal{P}}
\newcommand{\rscript}{\mathcal{R}}
\newcommand{\sscript}{\mathcal{S}}

\newcommand{\ab}[1]{\left|#1\right|}
\newcommand{\doubleab}[1]{\left|\left|#1\right|\right|}
\newcommand{\brac}[1]{\left\{#1\right\}}
\newcommand{\paren}[1]{\left(#1\right)}
\newcommand{\sqbrac}[1]{\left[#1\right]}
\newcommand{\elbows}[1]{{\left\langle#1\right\rangle}}
\newcommand{\floorbrac}[1]{{\lfloor#1\rfloor}}
\newcommand{\sqparen}[1]{{\left[#1\right)}}

\begin{document}

\title{FINITE-DIMENSIONAL\\CONVEX EFFECT ALGEBRAS}
\author{Stan Gudder\\ Department of Mathematics\\
University of Denver\\ Denver, Colorado 80208\\
sgudder@du.edu}
\date{}
\maketitle

\begin{abstract}
We first show that the convex effect algebras (CEA) approach to quantum mechanics is more general than the general probabilistic theories approach. We then restrict our attention to finite-dimension CEA's. After an introductory Section~1, we present basic definitions in Section~2. Section~3 studies convex subeffect algebras and observables. In Section~4 we consider strong CEA's and strong observables. We show that a CEA is strong if and only if it is classical. Informationally complete observables on classical CEA's are studied in Section~5. Section~6 considers quantum CEA's in Hilbert spaces.
\end{abstract}

\section{Introduction}  
Various types of stochastic theories have recently been important in studies of quantum mechanics and its generalization. The two types that we shall consider here are general probabilistic theories (GPT) \cite{bhss13,fghlapp,fhl18,gbca17,jp17} and convex effect algebras (CEA) \cite{ gud73,gp98,gpbb99}. The central role in these theories is played by the set of effects $\escript$ and the set of states $\sscript$. The effects correspond to yes-no measurements or experiments and the states correspond to preparation procedures that specify the initial conditions of the system being measured. Usually, each effect $a$ and state $s$ experimentally determine a probability $P(a,s)\in\sqbrac{0,1}$ that the effect $a$ occurs when the system has been prepared in the state $s$. Simple and physically motivated properties of $P(a,s)$ determine a mathematical structure for the sets $\escript$ and $\sscript$.

This structure is given in terms of an ordered linear space $V$ \cite{nam57,rock70}. The sets $\escript$ and $\sscript$ are then represented by certain subsets of $V$ and its dual space $V^*$ which we discuss in detail in Section~2. We shall show that GPT and CEA determine two different ways of viewing the pair $(V,V^*)$. From the GPT viewpoint, the set of states is considered basic and is described by a set
$\sscript\subseteq V$ while the set of effects is secondary with $\escript\subseteq V^*$. From the CEA viewpoint,
$\escript$ is considered basic with $\escript\subseteq V$ while $\sscript$ is secondary with $\sscript\subseteq V^*$. Roughly speaking, GPT and CEA are dual viewpoints. However, GPT results in a stronger structure involving an order-determining set of states, while CEA is more general. The two viewpoints are equivalent when this order-determining set of states condition holds. Since CEA is more general, we shall employ this viewpoint for the remainder of the paper. We also restrict our discussion to finite-dimensional spaces $V$. Although this is a strong restriction, it includes theories of quantum computation and quantum information \cite{hz12,nc00}.

Section~3 presents the basic definitions of the theory and compares the GPT and CEA viewpoints. In Section~3 we characterize convex subeffect algebras of a CEA. In Section~4 we consider strong CEA's and strong observables. We show that a CEA is strong if and only if it is classical. Informationally complete observables on classical CEA's are characterized in Section~5. Moreover, a necessary but not sufficient condition and a sufficient but not necessary condition for a pair of observables to be informationally complete is presented. Finally, Section~6 considers quantum CEA's in Hilbert space.There is some overlap between this work and that given in \cite{fghlapp}. We include this to make the present article self-contained.

\section{Basic Definitions}  
Let $V$ be a real linear space with zero 0. A subset $K$ of $V$ is a \textit{positive cone} if $\real ^+K\subseteq K$,
$K+K\subseteq K$ and $K\cap (-K)=\brac{0}$. For $x,y\in V$ we define $x\le y$ if $y-x\in K$. Then $\binle$ is a partial order on $V$ and we call $(V,K)$ an
\textit{ordered linear space} with positive cone $K$ \cite{nam57,rock70}. We say that $K$ is \textit{generating} if $V=K-K$. Let $u\in K$ with $u\ne 0$ and form the interval
\begin{equation*}
\sqbrac{0,u}=\brac{x\in K\colon x\le u}
\end{equation*}
For $x\in\sqbrac{0,u}$, we call $x'=u-x\in\sqbrac{0,u}$ the \textit{complement} of $x$. It is easy to check that $\sqbrac{0,u}$ is a convex subset of $V$ and $\lambda x\in\sqbrac{0,u}$ for all $\lambda\in\sqbrac{0,1}\subseteq\real$, $x\in\sqbrac{0,u}$. We say that $\sqbrac{0,u}$ is
\textit{generating} if $K=\real ^+\sqbrac{0,u}$ and $V=K-K$. If $\escript =\sqbrac{0,u}$ is generating, we call $\escript$ a \textit{convex effect algebra} (CEA). (It can be shown \cite{gpbb99,nam57} that $V$ is a normed space but this will not be needed if $V$ is finite-dimensional which we assume later.) For $a,b\in\escript$, if $a+b\in\escript$ we write
$a\perp b$. Then $\binperp$ and $\binle$ determine each other in the sense that $a\perp b$ if and only if $a\le b'$.

The \textit{dual} $V^*$ of $V$ is the set of (bounded) linear functionals $f\colon V\to\real$. We define
\begin{equation*}
V_+^*=\brac{f\in V^*\colon f(x)\ge 0 \hbox{ for all }x\in K}
\end{equation*}
Then $(V^*,V_+^*)$ becomes an ordered linear space called the \textit{dual} of $(V,K)$. A \textit{state} on $V$ is an element $s\in V_+^*$ satisfying $s(u)=1$ and we denote the set of states by $\sscript$. The elements of
$\escript =\sqbrac{0,u}$ represent effects, 0 is the effect that is always false (no) and $u$ is the effect that is always true (yes). If $a\in\escript$, $s\in\sscript$, then $s(a)$ gives the probability that $a$ is true in the state $s$. Of course, $s(a)\in\sqbrac{0,1}$ and $s(a+b)=s(a)+s(b)$ whenever $a\perp b$. We say that $\sscript$ is
\textit{order-determining} if $s(a)\le s(b)$ for every $s\in\sscript$ implies that $a\le b$. In general, $\sscript$ is not order-determining
\cite{gpbb99}. We call $(\escript ,\sscript )$ a CEA \textit{viewpoint} of a physical system. In this case, $\escript$ serves the primary role and
$\sscript$ is secondary.

To consider the GPT viewpoint, let $(V,K)$ again be an ordered linear space and let $u\in V_+^*$ with $u\ne 0$. In this case, the set of states $\sscript$ serves the primary role where
\begin{equation*}
\sscript =\brac{s\in K\colon u(s)=1}
\end{equation*}
Then $\sscript$ is a convex set which we can assume generates $K$ \cite{fghlapp,fhl18}. The set of effects
$\escript =\sqbrac{0,u}\subseteq V_+^*$ is now secondary and if $a\in\escript$, $s\in\sscript$, the $a(s)$ represents the probability that $a$ is true in the state $s$. If $a\in V^*$ then $a\in V_+^*$ if and only if $a(v)\ge 0$ for every
$v\in K$ which is equivalent to $a(s)\ge0$ for every $s\in\sscript$. Hence, if $a,b\in\escript$, then $a\le b$ if and only if $b-a\ge 0$. This is equivalent to $(b-a)(s)\ge 0$ for all $s\in\sscript$ which holds when $a(s)\le b(s)$. We conclude that $\escript$ has an order-determining set of states $\sscript$.

The main difference between the CEA and GPT viewpoints is that in the latter there is an order-determining set of states while in the former this need not hold. It can be shown that if $\sscript$ is order-determining on $\escript$, then the two viewpoints are equivalent, each being the dual of the other \cite{gpbb99}. In this case, $(\escript ,\sscript ,P)$ forms a \textit{effect-state space} where $P\colon\sscript\times\escript\to\sqbrac{0,1}$ given by $P(s,a)=s(a)$ is the probability function. We conclude that the CEA viewpoint is more general than the GPS viewpoint. For this reason, we shall employ the CEA viewpoint in the sequel.

We shall also assume that the linear space $V$ is finite-dimensional. Of course, this is a strong restriction, but it saves us from considering technical topological details. This finite-dimensional framework is strong enough to include the theory of quantum computation and quantum information which has been important recently \cite{hz12,nc00}.

\section{Convex Subeffect Algebras}  
In the sequel, we shall assume that $\escript =\sqbrac{0,u}$ is a CEA where $\sqbrac{0,u}$ is a generating interval in a finite-dimensional ordered linear space $V$. If $\dim V=n$, we define $\dim\escript =n$. A subset
$\fscript\subseteq\escript$ is a \textit{convex subeffect algebra} (CSEA) of $\escript$ if $0,u\in\fscript$, $a\in\fscript$ implies that $a'=u-a\in\fscript$, $\lambda a\in\fscript$ for every $a\in\fscript$,
$\lambda\in\sqbrac{0,1}\subseteq\real$ and $a,b\in\fscript$ with $a\perp b$ implies that $a+b\in\fscript$.

\begin{lem}    
\label{lem31}
Let $\fscript$ be a CSEA of $\escript$.
{\rm{(i)}}\enspace $\fscript$ is a convex subset of $\escript$ in the sense that $a,b\in\fscript$, $\lambda\in\sqbrac{0,1}$ imply that
$\lambda a+(1-\lambda )b\in\fscript$. 
{\rm{(ii)}}\enspace If $n\in\Natural$, $a\in\fscript$ and $na\le u$, then $na\in\fscript$.
{\rm{(iii)}}\enspace If $\lambda\in\real$ with $\lambda \ge 0$ and $a\in\fscript$, $\lambda a\le u$, then $\lambda a\in\fscript$.
{\rm{(iv)}}\enspace If $a,b\in\fscript$ and $b\le a$, then $a-b\in\fscript$. 
\end{lem}
\begin{proof}
(i)\enspace Since $\lambda a\le a$, $(1-\lambda )b\le b$ we have that $\lambda a,(1-\lambda )b\in\fscript$. Since
\begin{equation*}
\lambda a+(1-\lambda )b\le\lambda u+(1-\lambda )u=u
\end{equation*}
we have that $\lambda a\perp (1-\lambda )b$. Hence, $\lambda a+(1-\lambda )b\in\fscript$.\newline
(ii)\enspace We prove the result by induction on $n$. The result surely holds for $n=1$. Suppose the result holds for
$n\ge 1$ and
$(n+1)a\le u$. Then $na+a\le u$ and since $na\le u$ we have that $na\in\fscript$. Since $na\perp a$ and $a\in\fscript$ we have that
\begin{equation*}
(n+1)a=na+a\in\fscript
\end{equation*}
which proves the result by induction.
(iii)\enspace If $\lambda\le 1$, then $\lambda a\in\fscript$ so suppose that $\lambda >1$. Letting $\floorbrac{\lambda}$ be the integer part of
$\lambda$ and $\mu =\lambda -\floorbrac{\lambda}$ we have that $0\le\mu\le 1$ and $\lambda =\floorbrac{\lambda}+\mu$. Since
$\floorbrac{\lambda}a\le\lambda a$ and $\lambda a\in\escript$ we have that $\floorbrac{\lambda}a\in\escript$. By (ii) we have that
$\floorbrac{\lambda}a\in\fscript$ and $\mu a\in\fscript$. Since $\floorbrac{\lambda}a+\mu a=\lambda a\in\escript$ we have that $\floorbrac{\lambda}a\perp\mu a$. Hence, $\lambda a\in\fscript$.
(iv)\enspace Since $b\le (a')'$ we have that
\begin{equation*}
(a-b)'=u-(a-b)=b+(u-a)=b+a'\in\fscript
\end{equation*}
Hence, $a-b\in\fscript$.
\end{proof}

\begin{thm}    
\label{thm32}
Let $\escript =\sqbrac{0,u}\subseteq V$ be a CEA. Then $\fscript\subseteq\escript$ is a CSEA of $\escript$ if and only if there exists a linear subspace $V_1$ of $V$ such that $u\in V_1$ and $\fscript =\sqbrac{0,u}\cap V_1=\escript\cap V_1$.
\end{thm}
\begin{proof}
If $\fscript =\sqbrac{0,u}\cap V_1$ with $u\in V_1$, then clearly $\fscript$ is a CSEA of $\escript$. Conversely, let $\fscript$ be a CSEA of
$\escript$. Let $V_1$ be the subspace of $V$ generated by $\fscript$. Then $u\in V_1$ and we shall show that
$\fscript =\sqbrac{0,u}\cap V_1$. If $a\in\fscript$, then $a\in V_1$ and $a\in\escript$ so $a\in\sqbrac{0,u}\cap V_1$. Conversely, suppose $a\in\sqbrac{0,u}\cap V_1$. If $a=0$, then clearly $a\in\fscript$ and if $a\ne 0$, then $a=\sum\lambda _ia_i$, $a_i\in\fscript$,
$\lambda _i\in\real$, $\lambda _i\ne 0$. We can write
\begin{equation*}
a=\sum\alpha _ib_i-\sum\beta _ic_i
\end{equation*}
where $\alpha _i,\beta _i>0$, $b_i,c_i\in\fscript$. Letting $\lambda =\sum\alpha _i$, $\mu =\sum\beta _i$
\begin{equation*}
b=\sum\tfrac{\alpha _i}{\lambda}\,b_i,\quad c=\sum\tfrac{\beta _i}{\mu}\,c_i
\end{equation*}
we obtain $b,c\in\fscript$, $\lambda ,\mu >0$ and $a=\lambda b-\mu c$. Now $0\le\lambda b-\mu c\le u$ implies that
\begin{equation*}
\mu c\le\lambda b\le u+\mu c\le (1+\mu )u
\end{equation*}
Hence,
\begin{equation*}
0\le\tfrac{\mu}{1+\mu}\, c\le\tfrac{\lambda}{1+\mu}\, b\le u
\end{equation*}
It follows from Lemma~\ref{lem31}(iii) that $\tfrac{\lambda}{1+\mu}\, b\in\fscript$ and by Lemma~\ref{lem31}(iv) we have that
\begin{equation*}
d=\tfrac{\lambda}{1+\mu}\, b-\tfrac{\mu}{1-\mu}\, c\in\fscript
\end{equation*}
Since $(1+\mu )d=a\in\sqbrac{0,u}$, by Lemma~\ref{lem31}(iii) we obtain $a\in\fscript$. Hence, $\fscript =\sqbrac{0,u}\cap V_1$.
\end{proof}

Applying Theorem~\ref{thm32}, we conclude that $\fscript =\sqbrac{0,u}\subseteq V_1$ where $\sqbrac{0,u}$ generates $V_1$ and
$\dim\fscript =\dim V_1$. Hence, a CSEA is a CEA in its own right.

If $\fscript _1=\escript\cap V_1$, $\fscript _2=\escript\cap V_2$ are CSEA's of $\escript$, it is clear that
\begin{equation*}
\fscript _1\cap\fscript _2=\escript\cap V_1\cap V_2
\end{equation*}
is the largest CSEA contained in $\fscript _1$ and $\fscript _2$ and we write
$\fscript _1\wedge\fscript _2=\fscript _1\cap\fscript _2$. If
\begin{equation*}
\fscript _1\cap\fscript _2=\brac{a\in\escript\colon a=\lambda u,\lambda\in\sqbrac{0,u}}
\end{equation*}
we say that $\fscript _1$ and $\fscript _2$ are \textit{separated}. In this case, $\fscript _1\cap\fscript _2$ is isomorphic to $\sqbrac{0,1}$ and $V_1\cap V_2$ is isomorphic to $\real$. The smallest CSEA of $\escript$ containing
$\fscript _1$ and $\fscript _2$ is
\begin{equation*}
\fscript _1\vee\fscript _2=\escript\cap (V_1\vee V_2)
\end{equation*}
where $V_1\vee V_2$ is the subspace of $V$ generated by $V_1$ and $V_2$. If $\fscript _1$ and $\fscript _2$ are separated, then
\begin{equation*}
\dim (\fscript _1\vee\fscript _2)=\dim\fscript _1+\dim\fscript _2-1
\end{equation*}

\begin{cor}    
\label{cor33}
$\fscript$ is a CSEA of $\escript$ if and only if there exist linearly independent effects $a_1,a_2,\ldots$, $a_m\in\escript$ such that
$\sum r_ia_i=u$ for some $r_i\in\real$ and
\begin{equation}                
\label{eq31}
\fscript =\brac{a\in\escript\colon a=\sum\lambda _ia_i,\quad\lambda _i\in\real}
\end{equation}
\end{cor}
\begin{proof}
If $a_1,\ldots ,a_m\in\escript$ satisfying the given conditions, it is easy to verify that $\fscript$ is a CSEA of $\escript$. Conversely, if $\fscript$ is a CSEA of $\escript$, then by Theorem~\ref{thm32}, there exists a linear subspace $V_1$ of $V$ such that $u\in V_1$ and
$\fscript =\sqbrac{0,u}\cap V_1$. Let $v_1,v_2,\ldots ,v_r$ be a basis for $V_1$. Since $\fscript$ generates $V_1$, every $v_i$ has the form
$v_i=\alpha _iv_i^+-\beta _iv_i^-$, $v_i^+,v_i^-\in\fscript$, $\alpha _i,\beta _i\ge 0$. Let $b_i=v_i^+$, $i=1,2,\ldots ,r$ and
$b_i=v_{i-r}^-$, $i=r+1,r+2,\ldots ,2r$. Since $\brac{v_i}$ is a basis for $V_1$, we have that $\sum r_ib_i=u$ for some $r_i\in\real$. Also, since
$\brac{v_i}$ is a basis, \eqref{eq31} holds with $a_i$ replaced by $b_i$. Replacing $\brac{b_i}$ by a linearly independent subset
$\brac{a_1,a_2,\ldots ,a_m}$ we obtain $\sum r_ia_i=u$ for some $r_i\in\real$ and \eqref{eq31}. Since $\brac{a_i}$ generates $V_1$ we have that $m=r$.
\end{proof}
We conclude from the proof of Corollary~\ref{cor33} that $\dim\fscript =m$ and we call $a_1,a_2,\ldots ,a_m$ in Corollary~\ref{cor33}
\textit{generators} of $\fscript$. Although the $a_i$ are not unique, $m$ is unique.

An effect $a\in\escript$ is \textit{strong} if $a\not\le\lambda u$ for all $\lambda\in\sqparen{0,1}$. If we strengthen the properties of the generators in Corollary~\ref{cor33}, we obtain an interesting special type of CSEA.

\begin{lem}    
\label{lem34}
Let $a_1,a_2,\ldots ,a_m\in\escript$ be strong, linearly independent and satisfy $\sum a_i=u$. Then
\begin{equation}                
\label{eq32}
\fscript =\brac{a\in\escript\colon a=\sum\lambda _ia_i,\quad\lambda _i\in\sqbrac{0,1}}
\end{equation}
is a CSEA of $\escript$.
\end{lem}
\begin{proof}
The only condition that is not evident is that $a,b\in\fscript$ with $a\perp b$ implies $a+b\in\fscript$. We then suppose that $a,b\in\fscript$ with $a\perp b$ where $a=\sum\lambda _ia_i$, $b=\sum\mu _ia_i$, $\lambda _i,\mu _i\in\sqbrac{0,u}$. Then $a+b\in\escript$ and $a+b=\sum (\lambda _i+\mu _i)a_i$. Since
\begin{equation*}
(\lambda _i+\mu _i)a_i\le a+b\le u
\end{equation*}
we have that $a_i\le (\lambda _i+\mu _i)^\minusone u$ (we can assume $\lambda _i\ne 0$ or $\mu _i\ne 0$). Since $a_i$ is strong, we have that $(\lambda _i+\mu _i)^\minusone\ge 1$ so $\lambda _i+\mu _i\le 1$, $i=1,2,\ldots ,m$. Hence, $a+b\in\fscript$.
\end{proof}

\section{Strong Convex Effect Algebras}  
Motivated by Lemma~\ref{lem34} we say that a CEA $\escript$ is \textit{strong} if there exist a linearly independent set of effects $\brac{a_1,a_2,\ldots ,a_n}$ such that $\sum a_i=u$ and
\begin{equation*}
\escript =\brac{\sum\lambda _ia_i\colon\lambda _i\in\sqbrac{0,1}}
\end{equation*}
We see that the CEA $\fscript$ in Lemma~\ref{lem34} is strong. We call the $a_i$ in the previous definition \textit{generators} of $\escript$. We now show that the generators are automatically strong.

\begin{lem}    
\label{lem41}
If $\escript$ is a strong CEA with generators $\brac{a_i}$, then $a_i$ is strong for all $i$.
\end{lem}
\begin{proof}
Suppose that $a_j$ is not strong so that $a_j\le\lambda u$, $\lambda\in (0,1)$. Then
$\lambda ^\minusone a_j\in\escript$ and
$\lambda ^\minusone >1$. Let $0<\mu<1$ with $\mu <(1-\lambda )\lambda ^\minusone$. Then $a_j,\mu a_j\in\escript$ and since
$1+\mu <\lambda ^\minusone$ we have that
\begin{equation*}
a_j+\mu a_j=(1+\mu)a_j\le\lambda ^\minusone a_j
\end{equation*}
Hence, $a_j+\mu a_j\in\escript$ so we have that
\begin{equation*}
(1+\mu )a_j=\sum\lambda _ia_i,\quad\lambda _i\in\sqbrac{0,1}
\end{equation*}
Since representations are unique we conclude that $\lambda _j=1+\mu >1$ which gives a contradiction. Hence, $a_j$ is strong for all $j$.
\end{proof}

An effect $a$ in a CEA $\escript$ is \textit{sharp} if $a\wedge a'=0$ \cite{gud73}. That is, if $b\in\escript$ satisfies $b\le a,a'$ then $b=0$. It is clear that $0$ and $u$ are sharp. Physically, an effect is sharp if it is precisely yes or no when measured \cite{gud73}.

\begin{lem}    
\label{lem42}
If $a\ne 0$ is sharp, then $a$ is strong.
\end{lem}
\begin{proof}
If $a$ is not strong, then $a\le\lambda u$ for some $\lambda\in (0,1)$. Hence,
\begin{equation*}
a'\ge (\lambda u)'=(1-\lambda )u
\end{equation*}
Now $(1-\lambda )a\le a$ and $(1-\lambda )a\le (1-\lambda )u\le a'$. Since $(1-\lambda )a\ne 0$, $a$ is not sharp.
\end{proof}

An \textit{observable} on a CEA $\escript$ with \textit{finite outcome set} $X$ is a map $A\colon X\to\escript$ satisfying
\begin{equation}                
\label{eq41}
\sum _{x\in X}A(x)=u
\end{equation}
We sometimes write $A=\brac{A(x_1),\ldots ,A(x_n)}$ and interpret $A$ as a measurement with values $x_1,\ldots ,x_n$ such that $A(x_i)$ is the effect that occurs when $A$ has the value $x_i$. For example, the generators of a strong CEA form an observable
$A\colon\brac{1,\ldots ,n}\to\escript$ given by $A(i)=a_i$. The condition \eqref{eq41} says that $A$ must have one of the values $A(x)$,
$x\in X$. If $s$ is a state on $\escript$, then $s(A(x))\in\sqbrac{0,1}$ gives the probability that $A$ has the value $x$ when the system is in state $s$. Of course, this gives a probability measure because $\sum s(A(x))=s(u)=1$. An observable $A$ is \textit{strong} if $A(x)$ are linearly independent and strong. It follows from Lemma~\ref{lem41} that the generators of a strong CEA form a strong observable.

Two effects $a,b\in\escript$ \textit{coexist} if there exist effects $a_1,b_1,c\in\escript$ such that $a_1+b_1+c\in\escript$ and $a=a_1+c$, $b=b_1+c$ \cite{bhss13,hz12}. This terminology stems from the fact that we can then form the observable $A=\brac{a_1,b_1,c,d}$ where $d=(a_1+b_1+c)'$ and we can measure $a$ and $b$ simultaneously by measuring the single observable $A$.

\begin{lem}    
\label{lem43}
If $\escript$ is a strong CEA and $b,c\in\escript$, then $b$ and $c$ coexist.
\end{lem}
\begin{proof}
Let $\brac{a_i}$ be a set of generators for $\escript$. Then $b=\sum\lambda _ia_i$, $c=\sum\mu _ia_i$,
$\lambda _i,\mu _i\in\sqbrac{0,1}$. Define
\begin{equation*}
d=\sum\min (\lambda _i,\mu _i)a_i\in\escript
\end{equation*}
Then $b_1=b-d$, $c_1=c-d\in\escript$ and
\begin{align*}
b_1+c_1+d&=b+c-d=\sum\sqbrac{\lambda _i+\mu _i-\min (\lambda _i,\mu _i)}a_i\\
  &=\sum\max (\lambda _i,\mu _i)a_i\le u
\end{align*}
Hence, $b_1+c_1+d\in\escript$ and $b=b_1+d$, $c=c_1+d$. Therefore, $b$ and $c$ coexist.
\end{proof}

A \textit{classical channel} between outcome spaces $X$ and $Y$ is given by a stochastic matrix $\nu _{xy}$, $x\in X$, $y\in Y$ with $0\le\nu _{xy}\le 1$ and $\sum _{y\in Y}\nu _{xy}=1$ for all $x\in X$. We interpret $\nu _{xy}$ as the transition probability that outcome $x$ is mapped into outcome $y$ \cite{fghlapp,fhl18, gbca17}. For an observable
$A$ with outcome space $X$ and a classical channel $\nu$ from $X$ to $Y$, we define a new observable $\nu\circ A$ on $Y$ by
\begin{equation*}
(\nu\circ A)(y)=\sum _{x\in X}\nu _{xy}A(x)
\end{equation*}
for all $y\in Y$. Physically, $\nu\circ A$ is interpreted as first measuring $A$ and then employing the classical channel $\nu$ on each measurement outcome \cite{fghlapp,fhl18, gbca17}. For two observables $A$ and $B$, we say that
$B$ is a \textit{postprocessing} of $A$ denoted by $A\to B$ if there exists a classical channel $\nu$ such that
$B=\nu\circ A$.

\begin{thm}    
\label{thm44}
Let $A=\brac{a_i}$ be an observable with linearly independent elements $a_i$ in a CEA $\escript$. Then $\escript$ is strong with generators $a_i$ if and only if every observable in $\escript$ is a postprocessing of $A$.
\end{thm}
\begin{proof}
Let $\escript$ be strong with generators $A=\brac{a_1,a_2,\ldots ,a_n}$. We view $X=\brac{1,2,\ldots ,n}$ as an outcome space and write
\begin{equation*}
A=\brac{a_x\colon x\in X}=\brac{A(x)\colon x\in X}
\end{equation*}
Now let $B=\brac{B(y)\colon y\in Y}$ be another observable in $\escript$. Then $B(y)\in\escript$ so we have that $B(y)=\sum _{x\in X}\nu _{xy}A(x)$ where $0\le\nu _{xy}\le 1$. Since $B$ is an observable, we conclude that
\begin{equation*}
u=\sum _{y\in Y}B(y)=\sum _x\sqbrac{\sum _y\nu _{xy}}A(x)
\end{equation*}
Letting $\mu _x=\sum _{y\in Y}\nu _{xy}$ we have that $\sum _{x\in X}\mu _xA(x)=u$. If $\mu _x>1$ for some
$x\in X$, then since $\mu _xA(x)\le u$ we have that $A(x)\le (\mu _x)^\minusone u$. But this contradicts the fact that $A$ is strong. Hence, $\mu _x\le 1$ for all $x\in X$. If $\mu _x<1$ for some $x\in X$, then
\begin{equation*}
u=\sum _{x\in X}\mu _xA(x)<\sum _{x\in X}A(x)=u
\end{equation*}
which is a contradiction. Hence, $\sum _{y\in Y}\nu _{xy}=\mu _x=1$ for $x\in X$ so $A\to B$. Conversely, suppose every observable in $\escript$ is a postprocessing of $A$. If $a\in\escript$, then $B=\brac{a,a'}$ is an observable in
$\escript$ so $A\to B$. Hence, there exists $\nu _{ij}\in\sqbrac{0,1}$, $i=1,2,\ldots ,n$, $j=1,2$ such that
\begin{equation*}
a=B(1)=\sum\nu _{i1}a_i
\end{equation*}
We conclude that $\escript$ is strong with generators $\brac{a_i}$
\end{proof}

Let $\escript _1=\sqbrac{0,u_1}$, $\escript _2=\sqbrac{0,u_2}$ be CEA's. A \textit{morphism} from $\escript _1$ to
$\escript _2$ is a map $\phi\colon\escript _1\to\escript _2$ satisfying $\phi (u_1)=u_2$ and $a,b\in\escript _1$ with
$a\perp b$ implies that $\phi (a)\perp\phi (b)$ and $\phi (a+b)=\phi (a)+\phi (b)$. A morphism
$\phi\colon\escript _1\to\escript _2$ satisfying $\phi (a)\perp\phi (b)$ implies that $a\perp b$ is a \textit{monomorphism}. It is easy to check that a monomorphism is injective. Also, if $\phi$ is a surjective monomorphism, then
$\phi ^\minusone$ is a morphism and we call $\phi$ an \textit{isomorphism}. If $\phi$ is an isomorphism that satisfies
$\phi (\lambda a)=\lambda\phi (a)$ for all $\lambda\in\sqbrac{0,1}$, $a\in\escript _1$, then $\phi$ is an \textit{affine}
isomorphism and we say that $\escript _1$ and $\escript _2$ are \textit{affinely isomorphic} \cite{gpbb99}.

For $n\in\Natural$, let $\real ^n=\brac{(a_1,a_2,\ldots ,a_n)\colon a_i\in\real}$ be the real linear space with
\begin{equation*}
(a_1,a_2,\ldots ,a_n)+(b_1,b_2,\ldots ,b_n)=(a_1+b_1,a_2+b_2,\ldots ,a_n+b_n)
\end{equation*}
and $\lambda (a_1,a_2,\ldots ,a_n)=(\lambda a_1,\lambda a_2,\ldots ,\lambda a_n)$. Let
\begin{equation*}
K_n=\brac{(a_1,a_2,\ldots ,a_n)\colon a_i\ge 0,i=1,2,\ldots ,n}
\end{equation*}
be a positive cone in $\real ^n$. Letting $u_n=(1,1,\ldots ,1)$ we see that $S_n=\sqbrac{0,u_n}$ is a generating interval for the ordered linear space, $(\real ^n,K_n)$ so $S_n$ is a CEA. It is easy to verify that $a=(a_1,a_2,\ldots ,a_n)\in S_n$ is strong if and only if $a_i=1$ for some $i=1,2,\ldots ,n$ and $a$ is sharp if and only if $a_i=0$ or $1$ for all $i=1,2,\ldots ,n$. We say that a CEA $\escript$ is \textit{classical} if $\escript$ is affinely isomorphic to $S_n$ for some $n\in\Natural$.

\begin{thm}    
\label{thm45}
A CEA $\escript$ is strong if and only if $\escript$ is classical.
\end{thm}
\begin{proof}
Let $\escript$ be a strong CEA with generators $\brac{a_1,a_2,\ldots ,a_n}$. Define $J\colon\escript\to S_n$ by $J(a)=(\lambda _1,\lambda _2,\ldots ,\lambda _n)$ when $a=\sum\lambda _ia_i$. If $a,b\in\escript$ with $a\perp b$ where $a=\sum\lambda _ia_i,b=\sum\mu _ia_i$, then
\begin{equation*}
J(a+b)_i=\lambda _i+\mu _i\le 1
\end{equation*}
so $J(a)\perp J(b)$ and $J(a+b)=J(a+b)$. Also $J(u)=u_n$ and if $J(a)\perp J(b)$ then $a\perp b$. Moreover, we see that $J(\lambda a)=\lambda J(a)$ for all $\lambda\in\sqbrac{0,1}$, $a\in\escript$ and that $J$ is surjective. It follows that $J$ is an affine isomorphism so $\escript$ is classical. Conversely, suppose $\escript$ is a classical CEA and let $J\colon\escript\to S_n$ be an affine isomorphism. Let $\delta _i\in S_n$ be the element satisfying
$\delta _i(j)=\delta _{ij}$, $i,j=1,2,\ldots ,n$ and let $a_i=J^\minusone (\delta _i)\in\escript$. If $a\in\escript$, then there exists a $b\in S_n$ given by
\begin{equation*}
b=(\lambda _1,\lambda _2,\ldots ,\lambda _n)=\sum\lambda _i\delta _i
\end{equation*}
such that
\begin{equation*}
a=J^\minusone (b)=J^\minusone\paren{\sum\lambda _i\delta _i}=\sum\lambda _iJ^\minusone (b_i)
  =\sum\lambda _ia_i
\end{equation*}
Since $J$ is an isomorphism, this representation is unique. It follows that the $a_i$'s are linearly independent. Also,
\begin{equation*}
u=J^\minusone (u_n)=J^\minusone\paren{\sum\delta _i}=\sum J^\minusone (\delta _i)=\sum a_i
\end{equation*}
so $\brac{a_1,a_2,\ldots ,a_n}$ generates $\escript$. Hence, $\escript$ is a strong CEA.
\end{proof}

\section{Informationally Complete Random Variables}  
We now view the CEA $S_n$ of Section~4 in terms of classical probability theory. Let $X=\brac{1,2,\ldots ,n}$ and for every $a=(\lambda _1,\lambda _2,\ldots ,\lambda _n)$ in $S_n$ define the function $f_a\colon X\to\sqbrac{0,1}$ by $f_a(i)=\lambda _i$. We call $f_a$ a \textit{fuzzy event} and $a\mapsto f_a$ maps $S_n$ onto the set of fuzzy events $\fscript (X)$ on $X$. If $a,b\in S_n$ with $a\perp b$, we define $f_{a+b}(i)=f_a(i)+f_b(i)$ and for
$\lambda\in\sqbrac{0,1}$ we define $\lambda f_a=f_{\lambda a}$. Then $\fscript (X)$ becomes a CEA that is affinely isomorphic to $S_n$. If $s=(\mu _1,\mu _2,\ldots,\mu _n)$ is a state on $S_n$, we have the corresponding probability measure on $X$ given by $\mu _s(i)=\mu _i$. Denoting the set of states on $S_n$ by $\sscript _n$ and the set of probability measures on $X$ by $\pscript (X)$ we see that $(S_n,\sscript _n)$ and $\paren{\fscript (X),\pscript (X)}$ essentially coincide. In the literature, $\paren{\fscript (X),\pscript (X)}$ is called a \textit{fuzzy probability space}
\cite{bug96,gpbb99}.

We say that a strong CEA $\escript$ with generators $\brac{a_i}$ is \textit{sharp} if the $a_i$ are sharp,
$i=1,2,\ldots, n$. We can then identify $f_{a_i}$ with the set $\Gamma _i=\brac{j\in X\colon f_{a_i}(j)=1}$. Since
$\sum a_i=u$, we have that $\sum f_{a_i}=\chi _X$. It follows that $f_{a_i}f_{a_j}=0$ for $i\ne j$ so that
$\Gamma _i\cap\Gamma _j=\emptyset$ for $i\ne j$ and $\bigcup\Gamma _i=X$. Moreover, since
\begin{equation*}
\escript =\brac{\sum\lambda _ia_i\colon\lambda _i\in\sqbrac{0,1}}
\end{equation*}
we have that $f=\sum\lambda _if_{a_i}$ for all $f\in\fscript (X)$. Hence, we can assume without loss of generality that $f_{a_i}=\chi _{\brac{i}}$, $i=1,2,\ldots ,n$. The observable $A=\brac{a_i}$ on $S_n$ corresponds to the observable
$A$ on $\fscript (X)$ given by $\ahat (i)=\chi _{\brac{i}}$, $i=1,2,\ldots ,n$. We can thus identify $\ahat$ with the random variable $g_A$ on $X$ given by $g_A(i)=i$. By Theorem~\ref{thm44}, if $B$ is an observable on $\fscript (X)$, then $B$ is a postprocessing of $\ahat$ so we can represent $B$ by a random variable $g_B\colon X\to Y$ for some value space $Y$. In summary, if we have a sharp CEA, then we can represent the observables on $\escript$ by random variables on a set $X=\brac{1,2,\ldots ,n}$ and states on $\escript$ are represented by probability measures on $X$. This reduces the theory to classical probability.

Let $A=\paren{A(x_1),A(x_2),\ldots ,A(x_n)}$ be an observable on a CEA $\escript$. If $s$ is a state on $\escript$, then its \textit{probability distribution} $\Phi _{A,s}$ is given by
$\left\{s\sqbrac{A(x_1)},s\sqbrac{A(x_2)},\right.$ $\left.\ldots ,s\sqbrac{A(x_n)}\right\}$. We say that a collection of observables
$\brac{A_1,A_2,\ldots ,A_m}$ is \textit{informationally complete} if for any states
$s_1,s_2,\Phi _{A_i,s}=\Phi _{A_i,s_2}$, $i=1,2,\ldots ,m$, implies that $s_1=s_2$. Single informationally complete observables in Hilbert space quantum mechanics are well understood \cite{bus91,hz12}. However, this is not true for larger sets of observables. For example, we would like to characterize pairs of observables $(A_1,A_2)$ that are informationally complete where $A_1$ and $A_2$ are not.

In this section, we consider informationally complete observables in sharp CEA's. Of course, this is a very strong restriction, but it may give some insights for the general case. We then have the following situation. Let $X=\brac{1,2,\ldots ,n}$ and let $\pscript (X)$ be the set of probability measures on $X$. Every $\mu\in\pscript (X)$ has the form $\mu =(\mu _1,\mu _2,\ldots ,\mu _n)$ where $\mu _i\in\sqbrac{0,1}$, $\sum\mu _i=1$. Let $\rscript (X)$ be the set of random variables $f\colon X\to\real$. For $f\in\rscript (X)$, $\mu\in\pscript (X)$ the \textit{probability distribution} is the measure on sets $\Delta\subseteq\real$ given by
\begin{equation*}
\Phi _{f,\mu}(\Delta )=\mu\sqbrac{f^\minusone (\Delta )}=\sum\brac{\mu _i\colon f(i)\in\Delta}
\end{equation*}
We say that $f\in\rscript (X)$ is \textit{informationally complete} (IC) if $\Phi _{f,\mu}=\Phi _{f,\nu}$ implies $\mu =\nu$. We use the notation $\ab{X}=n$ and assume that $n\ge 2$.

\begin{exam}{1}  
We show that when $\ab{X}=2$, then $f$ is IC if and only if $f$ is not constant. If $f$ is constant, then $f(x_i)=\lambda$, $i=1,2$ and we have that
\begin{equation*}
\Phi _{f,\mu}\paren{\brac{\lambda}}=\Phi _{f,\nu}\paren{\brac{\lambda}}=1
\end{equation*}
It follows that $\Phi _{f,\mu}=\Phi _{f,\nu}$ for all probability measures $\mu$ and $\nu$ so $f$ is not IC. If $f$ is not constant, then $f(1)\ne f(2)$ and we have that $\Phi _{f,\mu}\paren{\brac{f(1)}}=\mu _1$,
$\Phi _{f,\mu}\paren{\brac{f(2)}}=\mu _2$ for any $\mu\in\pscript (X)$. If $\nu\in\pscript (X)$ is another probability measure with $\nu =(\nu _1,\nu _2)$ and $\Phi _{f,\mu}=\Phi _{f,\nu}$ then
\begin{align*}
\mu _1&=\Phi _{f,\mu}\paren{\brac{f(1)}}=\Phi _{f,\nu}\paren{\brac{f(1)}}=\nu _1\\
\mu _2&=\Phi _{f,\mu}\paren{\brac{f(2)}}=\Phi _{f,\nu}\paren{\brac{f(2)}}=\nu _2
\end{align*}
Hence, $\mu =\nu$ so $f$ is IC.\hfill\qedsymbol
\end{exam}

In general we have the following result.

\begin{thm}    
\label{thm51}
A random variable $f\in\rscript (X)$ is IC if and only if $f$ is injective.
\end{thm}
\begin{proof}
If $f(i)\ne f(j)$ for all $i\ne j$ and $\Phi _{f,\mu}=\Phi _{f,\nu}$ for $\mu ,\nu\in\pscript (X)$ then
\begin{equation*}
\mu _i=\Phi _{f,\mu}\paren{\brac{f(i)}}=\Phi _{f,\nu}\paren{\brac{f(i)}}=\nu _i
\end{equation*}
for all $i\in\brac{1,2,\ldots ,n}$. Hence, $\mu =\nu$ so $f$ is IC. Conversely, suppose $f$ is not injective and assume without loss of generality that $f(1)=f(2)$. Let $\mu =(1/2,1/2,0,\ldots ,0)$, $\nu =(1/4,3/4,0,\ldots ,0)$. Then
$\mu\ne\nu$ but
\begin{equation*}
\Phi _{f,\mu}\paren{\brac{f(1)}}=\Phi _{f,\nu}\paren{\brac{f(1)}}=1
\end{equation*}
so $\Phi _{f,\mu}=\Phi _{f,\nu}$. Hence, $f$ is not IC.
\end{proof}

A set of two random variables $\brac{f,g}$ is IC if $\Phi _{f,\mu}=\Phi _{f,\nu}$ and $\Phi _{g,\mu}=\Phi _{g,\nu}$ imply that $\mu =\nu$. Of course, if $f$ is IC, then $\brac{f,g}$ is IC for $g\in\rscript (X)$. The interesting case is when $\brac{f,g}$ is IC and neither $f$ nor $g$ is IC.

\begin{exam}{2}  
We show that when $\ab{X}=2$, then $\brac{f,g}$ is IC if and only if either $f$ or $g$ is IC. Indeed, if either $f$ or $g$ is IC then of course $\brac{f,g}$ is IC. Conversely, suppose both $f$ and $g$ are not IC. Then by Example~1, $f$ and $g$ are constant so $f(1)=f(2)$ and $g(1)=g(2)$. Let $\mu =(\mu _1,\mu _2)$, $\nu =(\nu _1,\nu _2)$ be probability measures. Then
\begin{align*}
\Phi _{f,\mu}\paren{\brac{f(1)}}&=\Phi _{f,\nu}\paren{\brac{f(1)}}=1\\
    \intertext{and}
\Phi _{g,\mu}\paren{\brac{g(1)}}&=\Phi _{g,\nu}\paren{\brac{g(1)}}=1
\end{align*}
Hence, $\Phi _{f,\mu}=\Phi _{f,\nu}$ and $\Phi _{g,\mu}=\Phi _{g,\nu}$ but $\mu\ne\nu$ in general so $\brac{f,g}$ is not IC.\hfill\qedsymbol
\end{exam}

\begin{exam}{3}  
Let $\ab{X}=3$ and let $f,g\in\rscript (X)$ where neither $f$ nor $g$ is IC. By Theorem~\ref{thm51}, $f$ and $g$ are not injective. If $f(1)=f(2)$ and $g(1)=g(2)$ then $\brac{f,g}$ is not IC. Indeed, let $\mu ,\nu\in\pscript (X)$ with
$\mu _1\ne\nu _1$ but $\mu _1+\mu _2=\nu _1+\nu _2$. We have that
\begin{align*}
\Phi _{f,\mu}\paren{\brac{f(1)}}&=\mu _1+\mu _2=\Phi _{g,\mu}\paren{\brac{g(1)}}\\
\Phi _{f,\nu}\paren{\brac{f(1)}}&=\nu _1+\nu _2=\Phi _{g,\nu}\paren{\brac{g(1)}}\\
\Phi _{f,\mu}\paren{\brac{f(3)}}&=\mu _3=\Phi _{g,\mu}\paren{\brac{g(3)}}\\
\Phi _{f,\nu}\paren{\brac{f(3)}}&=\nu _3=\Phi _{g,\nu}\paren{\brac{g(3)}}
\end{align*}
Since $\mu _3=\nu _3$ we have that $\Phi _{f,\mu}=\Phi _{f,\nu}$ and $\Phi _{g,\mu}=\Phi _{g,\nu}$ but $\mu\ne\nu$. Hence, $\brac{f,g}$ is not IC. On the other hand if $f(1)=f(2)\ne f(3)$ and $g(1)\ne g(2)=g(3)$, then $\brac{f,g}$ is IC. In this case, for $\mu ,\nu\in\pscript (X)$ we have that
\begin{align*}
\Phi _{f,\mu}\paren{\brac{f(1)}}&=\mu _1+\mu _2,\quad\Phi _{f,\mu}\paren{\brac{f(3)}}=\mu _3\\
\Phi _{f,\nu}\paren{\brac{f(1)}}&=\nu _1+\nu _2,\quad\Phi _{f,\nu}\paren{\brac{f(3)}}=\nu _3\\
\Phi _{g,\mu}\paren{\brac{g(1)}}&=\mu _1,\quad\Phi _{g,\mu}\paren{\brac{g(2)}}=\mu _2+\mu _3\\
\Phi _{g,\nu}\paren{\brac{g(1)}}&=\nu _1,\quad\Phi _{g,\nu}\paren{\brac{g(2)}}=\nu _2+\nu _3
\end{align*}
If $\Phi _{f,\mu}=\Phi _{f,\nu}$ and $\Phi _{g,\mu}=\Phi _{g,\nu}$ then $\mu _1+\mu _2=\nu_1+\nu _2$,
$\mu _3=\nu _3$ and $\mu _1=\nu _1$, $\mu _2+\mu _3=\nu _2+\nu _3$. Hence, $\mu =\nu$ so $\brac{f,g}$ is IC.
\hfill\qedsymbol
\end{exam}

A random variable $f\in\rscript (X)$ gives a partition of $X=\brac{1,2,\ldots ,n}$ where $\brac{i}$ is a singleton if
$f(i)\ne f(j)$ for any $j\ne i$, $\brac{i,j}$ is a doubleton if $f(i)=f(j)$ and $f(i)\ne f(k)$ for any $k\ne i,j$, etc. We denote the partition for $f$ by $P(f)$. If $A$ and $B$ are partitions of $X$, then the set of intersections of sets in $A$ with sets in $B$ (omitting the empty set) is denoted by $A\cap B$. We say that two random variables $f,g\in\rscript (X)$ are \textit{complementary} if $P(f)\cap P(g)$ consists of singleton sets.

\begin{exam}{4}  
Let $f,g$ be the first two random variables in Example~3. Then $P(f)=\brac{\brac{1,2},\brac{3}}$,
$P(g)=P(f)=P(f)\cap P(g)$ and $f,g$ are not complementary. Next, let $f,g$ be the second two random variables in Example~3. Then $P(f)=\brac{\brac{1,2},\brac{3}}$, $P(g)=\brac{\brac{1},\brac{2,3}}$ and we have
\begin{equation*}
P(f)\cap P(g)=\brac{\brac{1},\brac{2},\brac{3}}
\end{equation*}
so $f,g$ are complementary. Recall that in the first case, $\brac{f,g}$ were not IC while in the second case $\brac{f,g}$ were IC. As we shall see, this is no accident. As another example, let $h_1\in\rscript (X)$, where $\ab{X}=5$, with different values $h_1(1),h_1(2),\ldots ,h_1(5)$ except $h_1(1)=h_1(5)$ and $h_1(3)=h_1(4)$. We then have
\begin{equation*}
P(h_1)=\brac{\brac{1,5},\brac{2},\brac{3,4}}
\end{equation*}
If $h_2\in\rscript (X)$ satisfies $h_2(2)=h_2(3)=h_2(4)$ and the other values are different we have
\begin{equation*}
P(h_2)=\brac{\brac{1},\brac{2,3,4},\brac{5}}
\end{equation*}
Hence, $P(h_1)\cap P(h_2)=\brac{\brac{1},\brac{2},\brac{3,4},\brac{5}}$ so $h_1,h_2$ are not complementary.
\hfill\qedsymbol
\end{exam}

We say that $f,g\in\rscript (X)$ are \textit{strongly complementary} if for all $i\in X$, either $\brac{i}\in P(f)$ or
$\brac{i}\in P(g)$. For instance, none of the pairs of random variables in Example~4 are strongly complementary. An example of a strongly complementary pair is given by the partitions
\begin{equation*}
P(f)=\brac{\brac{1,2},\brac{3},\brac{4}},\quad P(g)=\brac{\brac{1}\brac{2},\brac{3,4}}
\end{equation*}
It is easy to check that a pair that is strongly complementary must be complementary. Also, $f,g\in\rscript (X)$ are strongly complementary if and only if for all $i\in X$ either $f(i)\ne f(j)$ or $g(i)\ne g(j)$ for all $j\in X$ with $j\ne i$.

\begin{thm}    
\label{thm52}
{\rm{(i)}}\enspace If $f,g\in\rscript (X)$ are strongly complementary, then $f,g$ are IC.
{\rm{(ii)}}\enspace If $f,g$ are IC, then $f,g$ are complementary.
\end{thm}
\begin{proof}
(i)\enspace Suppose $f,g\in\rscript (X)$ are strongly complementary. Let $\mu =(\mu _1,\mu _2,\ldots ,\mu _n)$ and
$\nu =(\nu _1,\nu _2,\ldots ,\nu _n)$ be states in $\pscript (X)$ and suppose that $\Phi _{f,\mu}=\Phi _{f,\nu}$ and
$\Phi _{g,\mu}=\Phi _{g,\nu}$. For $i\in X$, either $\brac{i}\in P(f)$ or $\brac{i}\in P(g)$. In the former case, we have
\begin{equation*}
\Phi _{f,\mu}\paren{\brac{f(i)}}=\mu _i=\Phi _{f,\nu}\paren{\brac{f(i)}}=\nu _i
\end{equation*}
while in the latter case, we have
\begin{equation*}
\Phi _{g,\mu}\paren{\brac{g(i)}}=\mu _i=\Phi _{g,\nu}\paren{\brac{g(i)}}=\nu _i
\end{equation*}
In this way, $\mu _i=\nu _i$, $i=1,2,\ldots ,n$ so $\mu =\nu$.\newline
(ii)\enspace Suppose that $f,g$ are not complementary. Without loss of generality, we can assume that $f(1)=f(2)$, $g(1)=g(2)$ while the other values of $f$ and $g$ are arbitrary. Let $\mu =(1,0,\ldots ,0)$ and $\nu = (0,1,0,\ldots ,0)$ be states on $X$. Then
\begin{align*}
\Phi _{f,\mu}\paren{\brac{f(1)}}&=1=\Phi _{f,\nu}\paren{\brac{f(1)}}\\
    \intertext{and}
\Phi _{g,\mu}\paren{\brac{g(1)}}&=1=\Phi _{g,\nu}\paren{\brac{g(1)}}
\end{align*}
But $\mu\ne\nu$ so $f,g$ are not IC.
\end{proof}

The definitions and Theorem~\ref{thm52} extend to more than two random variables in a natural way. In the second illustration of Example~3, $f,g$ are IC but $f,g$ are not strongly complementary. This shows that the converse of Theorem~\ref{thm52}(i) is false. Hence, strong complementary is a sufficient but not necessary condition for IC.

\begin{exam}{5}  
This example shows that the converse of Theorem~\ref{thm52}(ii) is false. Let $X=\brac{1,2,3,4}$ and suppose
$f,g\in\rscript (X)$ with
\begin{equation*}
P(f)=\brac{\brac{1,2},\brac{3,4}},\quad P(g)=\brac{\brac{1,3},\brac{2,4}}
\end{equation*}
Then $f,g$ are complementary. To show that $f,g$ are not IC, consider the distinct states
\begin{equation*}
\mu =\paren{\tfrac{1}{4},\tfrac{1}{4},\tfrac{1}{4},\tfrac{1}{4}},\quad
\nu =\paren{\tfrac{1}{3},\tfrac{1}{6},\tfrac{1}{6},\tfrac{1}{3}}
\end{equation*}
We then have that
\begin{align*}
\Phi _{f,\mu}\paren{\brac{f(1)}}&=\tfrac{1}{4}+\tfrac{1}{4}=\tfrac{1}{2}=\tfrac{1}{3}+\tfrac{1}{6}
  =\Phi _{f,\nu}\paren{\brac{f(1)}}\\
\Phi _{f,\mu}\paren{\brac{f(3)}}&=\tfrac{1}{4}+\tfrac{1}{4}=\tfrac{1}{2}=\tfrac{1}{6}+\tfrac{1}{3}
  =\Phi _{f,\nu}\paren{\brac{f(3)}}\\
\Phi _{g,\mu}\paren{\brac{g(1)}}&=\tfrac{1}{4}+\tfrac{1}{4}=\tfrac{1}{2}=\tfrac{1}{3}+\tfrac{1}{6}
  =\Phi _{g,\nu}\paren{\brac{g(1)}}\\
\Phi _{g,\mu}\paren{\brac{g(2)}}&=\tfrac{1}{4}+\tfrac{1}{4}=\tfrac{1}{2}=\tfrac{1}{6}+\tfrac{1}{3}
  =\Phi _{g,\nu}\paren{\brac{g(2)}}
\end{align*}
Hence, $f,g$ are not IC.\hfill\qedsymbol
\end{exam}

We conclude from Example~5 that complementarity is a necessary but not a sufficient condition for IC. It is an open problem to find a simple characterization of IC for a pair $\brac{f,g}\subseteq\rscript (X)$.

\section{Quantum Convex Effect Algebras}  
In this section we briefly consider CEA's on a Hilbert space. A more complete discussion is given in \cite{fghlapp}. Let $H$ be a finite-dimensional complex Hilbert space and let $\lscript _S(H)$ be the real linear space of self-adjoint operators on $H$. We can order the elements of $\lscript _S(H)$ using the cone of positive operators
$K\subseteq\lscript _S(H)$. Letting $I$ be the identity operator, we construct the CEA, $\escript (H)=\sqbrac{0,I}$ where $\sqbrac{0,I}$ is a generating interval in $\lscript _S(H)$. We call $\escript (H)$ a \textit{full quantum} CEA and the elements of $\escript (H)$ are called \textit{quantum effects}. The states on $\escript (H)$ are precisely the density operators on $H$; that is, the operators $\rho\in K$ with $\rmtr (\rho )=1$.  We then have that $\rho (a)=\rmtr (\rho a)$ for all $a\in\escript (H)$. Any CSEA of $\escript (H)$ is called a \textit{quantum} CSEA. It can be shown that
$a\in\escript (H)$ is sharp if and only if $a$ is a projection \cite{gpbb99}. We denote the spectrum of $a\in\escript (H)$ by
$\sigma (a)$.

\begin{lem}    
\label{lem61}
A quantum effect $a\in\escript (H)$ is strong if and only if $1\in\sigma (a)$.
\end{lem}
\begin{proof}
Suppose that $1\in\sigma (a)$. By the spectral theorem $a=p+b$ where $p$ is a one-dimensional projection and
$b\in\escript (H)$. If $a$ is not strong, then $a\le\lambda I$, $\lambda\in\sqparen{0,1}$. Hence, $p\le a\le\lambda I$. Let $\phi$ be a unit eigenvector of $p$ with corresponding eigenvalue $1$ so that $p\phi =\phi$. Then
\begin{equation*}
1=\elbows{\phi ,p\phi}\le\lambda\elbows{\phi ,\phi}=\lambda
\end{equation*}
which is a contradiction. Hence, $a$ is strong. Conversely, suppose that $a\in\escript (H)$ is strong. If $1\not\in\sigma (a)$, then $\doubleab{a}<1$. Since $a\le\doubleab{a}I$, this gives a contradiction. Hence $\in\sigma (a)$.
\end{proof}

It follows from Lemma~\ref{lem61} that strong effects need not be sharp.

It is not hard to show that if $\dim H=n$, then $\dim\lscript _S(H)=n^2$. Then for any $m\le n^2$ we can construct a CSEA $\fscript\subseteq\escript (H)$ with $\dim\fscript =m$. We say that a quantum CSEA is \textit{commutative} if all its elements commute. Of course, $\fscript$ is commutative if and only if its generators mutually commute. It is also clear, any full CEA is noncommutative. If a quantum CSEA $\fscript$ satisfies $\dim\fscript =2$, then $\fscript$ is commutative. This is because,its generators $a_1,a_2$ satisfy $r_1a_1+r_2a_2=I$ for some $r_1,r_2\in\real$ which implies $a_1a_2=a_2a_1$. We now give an example of a 3-dimensional noncommutative quantum CSEA.

\begin{exam}{6}  
Let $\alpha ,\beta\in\escript (\complex ^2)$ satisfy $\alpha\beta\ne\beta\alpha$ and
$0\not\in\sigma (\alpha ),\sigma (\beta )$. Letting $a_1=\tfrac{\alpha}{2}$, $a_2=\tfrac{\beta}{2}$,
$a_3=I-\tfrac{\alpha}{2}-\tfrac{\beta}{2}$ we have that $a_1,a_2,a_3\in\escript (\complex ^2)$ and $a_1+a_2+a_3=I$ so $A=\brac{a_1,a_2,a_3}$ is an observable. It is easy to check that  the $a_i$'s do not commute and are linearly independent. Hence, the quantum CSEA generated by $A$ is noncommutative. Notice that
$0,1\not\in\sigma (a_1),\sigma (a_2)$. If $0\in\sigma (a_3)$, then there exists a unit vector $\phi\in\complex ^2$ such that $\tfrac{1}{2}\,\elbows{\phi ,\alpha\phi}+\tfrac{1}{2}\elbows{\phi ,\beta\phi}=1$. But then
$\elbows{\phi ,\alpha\phi}=\elbows{\phi ,\beta\phi}=1$. This implies that $1\in\sigma (a_1)$ which is a contradiction. If $1\in\sigma (a_3)$, then there exists a unit vector $\psi\in\complex ^2$ such that
$\tfrac{1}{2}\,\elbows{\phi ,\alpha\phi}+\tfrac{1}{2}\,\elbows{\phi ,\beta\phi}=0$. As before, this implies that
$0\in\sigma (a_1)$ which is a contradiction. We conclude that $0,1\not\in\sigma (a_3)$ so $a_1,a_2,a_3$ are not strong.\hfill\qedsymbol
\end{exam}

The next result characterizes the strong quantum CSEA's.

\begin{thm}    
\label{thm62}
Let $a_1,\ldots ,a_m$ be generators for a strong CSEA $\fscript\subseteq\escript (H)$ where $\dim H=n$. Then
$m\le n$, there exist nonzero projections $P_i$, $i=1,\ldots ,m$ and a projection $Q$ with $P_1+\cdots +P_m+Q=I$ such that $a_i=P_i+Qa_iQ$ and $0,1\not\in\sigma (Qa_iQ)$.
\end{thm}
\begin{proof}
Let $P_i$ be the projections onto the eigenspace $\brac{\phi\in H\colon a_i\phi =\phi}$. Since $1\in\sigma (a_i)$,
$P_i\ne 0$. Suppose $a_i\phi =\phi$ where $\phi\ne 0$. Since $\sum a_k=1$ we have
\begin{equation*}
\phi =\sum a_k\phi =a_i\phi +\sum _{k\ne i}a_k\phi =\phi +\sum _{k\ne i}a_k\phi
\end{equation*}
Hence, $\sum _{k\ne i}a_k\phi =0$ so that $\sum _{k\ne i}\elbows{\phi ,a_k\phi}=0$. Since
$\elbows{\phi ,a_k\phi}\ge 0$ we obtain $\elbows{\phi ,a_k\phi}=0$ for all $k\ne i$. Thus
$\elbows{a_k^{1/2}\phi ,a_k^{1/2}\phi}=0$ so that $a_k^{1/2}\phi =0$ and we have that $a_k\phi =0$. If $k\ne i$ and
$a_i\psi =\psi$, then by the above $a_i\psi =0$. But $a_i\phi=\phi$ so $\psi$ and $\phi$ are eigenvectors with different eigenvalues. Hence, $\phi\perp\psi$. This implies that $P_iP_j=P_jP_i=0$ whenever $i\ne j$. Let $Q$ be the projection given by $Q=I-\sum _{i=1}^mP_i$ so that $\sum P_i+Q=I$. Then $QP_i=P_iQ=0$ for $i=1,\ldots ,m$. By the Spectral Theorem $a_i=P_i+b_i$ where $b_i$ is an effect with $0,1\not\in\sigma (b_i)$. Since
\begin{equation*}
P_ia_i\phi =a_iP_i\phi =P_i\phi
\end{equation*}
for all $\phi\in H$ we have that $P_ia_i=P_i$. Hence,
\begin{equation*}
a_i-P_i=\paren{\sum P_j+Q}(a_i-P_i)=P_ia_i+Qa_i-P_i=Qa_i
\end{equation*}
We conclude that
\begin{equation*}
b_i=a_i-P_i=Qa_i=Qa_iQ\qedhere
\end{equation*}
\end{proof}

It follows from Theorem~\ref{thm62} that if $S\subseteq\escript (H)$ is a strong CSEA with $\dim H=n$ then
$\dim S\le n$. Moreover, if $\dim S=n$ then there are one-dimensional projections $P_1,\ldots ,P_n$ with
$\sum P_i\!=\!1$ and $S\!=\!\brac{\sum\lambda _iP_i\colon\lambda_i\in\sqbrac{0,1}}$. We now give an example of a strong noncommutative quantum CSEA $\fscript$. This is surprising because by Theorem~\ref{thm45} we know that $\fscript$ must be classical.

\begin{exam}{7}  
Let $\dim H=5$ and let $\fscript\subseteq\escript (H)$ be a strong CSEA with $\dim\fscript =3$. If $a_1,a_2,a_3$ are generators of $\fscript$, it follows from Theorem~\ref{thm62} that there exist nonzero projections $P_1,P_2,P_3$ and a projection $Q$ such that $P_1+P_2+P_3+Q=I$ and $a_i=P_i+Qa_iQ$, $0,1\not\in\sigma (Qa_iQ)$. We can and will assume that $\dim Q=2$ from which it follows that $\dim P_i=1$, $i=1,2,3$. Since the $P_i$ and $Q$ mutually commute, they can be simultaneously diagonalized and writing the $a_i$ as matrices we have
\begin{align*}
a_1&=
\left[
\begin{blockarray}{@{\,}ccccc@{\:}}
\bigstrut[t]
1 & 0 & 0 & 0 & 0 & \\
0 & 0 & 0 & 0 & 0 & \\
0 & 0 & 0 & 0 & 0 & \\
\begin{block}{@{\,}ccc[\BAmulticolumn{2}{!{}c!{}}@{\:}]}
0 & 0 & 0 & \multirow{2}{*}{b}\\
0 & 0 & 0 & \\
\end{block}
\BAnoalign{\vskip -7ex}
\end{blockarray}\right]\quad
a_2=
\left[
\begin{blockarray}{@{\,}ccccc@{\:}}
\bigstrut[t]
0& 0 & 0 & 0 & 0 & \\
0 & 1 & 0 & 0 & 0 & \\
0 & 0 & 0 & 0 & 0 & \\
\begin{block}{@{\,}ccc[\BAmulticolumn{2}{!{}c!{}}@{\:}]}
0 & 0 & 0 & \multirow{2}{*}{c}\\
0 & 0 & 0 & \\
\end{block}
\BAnoalign{\vskip -7ex}
\end{blockarray}\right]\\
\noalign{\smallskip}
a_3&=
\left[
\begin{blockarray}{@{\,}ccccc@{\:}}
\bigstrut[t]
0& 0 & 0 & 0 & 0 & \\
0 & 0 & 0 & 0 & 0 & \\
0 & 0 & 1 & 0 & 0 & \\
\begin{block}{@{\,}ccc[\BAmulticolumn{2}{!{}c!{}}@{\:}]}
0 & 0 & 0 & \multirow{2}{*}{d}\\
0 & 0 & 0 & \\
\end{block}
\BAnoalign{\vskip -7ex}
\end{blockarray}\right]
\end{align*}
where $\sqbrac{b},\sqbrac{c},\sqbrac{d}\in\escript (\complex ^2)$ satisfy $\sqbrac{b}+\sqbrac{c}+\sqbrac{d}=I$ and $0,1\not\in\sigma\paren{\sqbrac{b}},\break \sigma\paren{\sqbrac{c}},\sigma\paren{\sqbrac{d}}$. Except for satisfying the above conditions, the effects $\sqbrac{b}$, $\sqbrac{c}$, $\sqbrac{d}$ are arbitrary and we can choose them to be noncommutative as in Example~6. Then $a_1,a_2,a_3$ do not commute so $\fscript$ is a noncommutative strong quantum CSEA. It is not hard to show that $\dim H=5$, $\dim\fscript =3$ are the smallest dimensions for such an example.\hfill\qedsymbol
\end{exam}

Example~7 shows that the converse of the next theorem is false.

\begin{thm}    
\label{thm63}
If a quantum CSEA $\fscript$ is commutative, then $\fscript$ is strong.
\end{thm}
\begin{proof}
Let $\fscript\subseteq\escript (H)$ be commutative with $\dim\fscript =m$ and $\dim H=n$. Then $\fscript$ has $m$ generators $a_1,\ldots ,a_m$ where $a_1,\ldots ,a_m$ mutually commute and are linearly independent. It follows that the $a_i$ are simultaneously diagonalizable so we can assume without loss of generality that $a_1,\ldots ,a_m$ are diagonal $n\times n$ matrices $a_i=\rmdiag (a_i^j)$, $i=1,\ldots ,m$, $j=1,\ldots ,n$, where $a_i^j\in\sqbrac{0,1}$. Since $a_1,\ldots ,a_m$ are linearly independent, they span an $m$-dimensional subspace $V$ of the real linear space
$\real ^n$. For $b\in V$ we denote the $j$th component of $b$ by $b^j$. We conclude that
$V=\brac{\sum\mu _ia_i\colon\mu _i\in\real}$ and $\fscript =\brac{b\in V\colon b^j\in\sqbrac{0,1}}$. It follows that $V$ is isomorphic to $\real ^m$ and $\fscript$ is isomorphic to the classical CEA $S_m$ via the map $J(b)(j)=b^j$, $j=1,\ldots ,m$. Applying Theorem~\ref{thm45}, we conclude that $\fscript$ is strong.
\end{proof}


\begin{thebibliography}{99}
\bibitem{bug96}S.~Bugajski, Fundamentals of fuzzy probability theory,
\textit{Int. J. Theor. Phys.} \textbf{35}, 2229 (1996).
\bibitem{bus91}P.~Busch, Informationally complete sets of physical quantities, \textit{Int. J. Theor. Phys.} \textbf{30}, 1217 (1991).
\bibitem{bhss13}P.~Busch, T.~Heinosaari, J.~Schultz and N.~Stevens, Comparing the degrees of incompatibility inherent in probabilistic physical theories, \textit{Europhys. Lett.} \textbf{103}, 10002 (2013).
\bibitem{dp94}A.~Dvuren\v censkij and S.~Pulmannov\'a, Difference posets, effects, and quantum measurements.
\textit{Int. J. Theor. Phys.} \textbf{33}, 819 (1994).
\bibitem{fghlapp}S.~Filippov, S.~Gudder, T.~Heinosaari and L.~Lepp\" aj\" arvi, Operational restrictions in general probabilistic theories, to appear.
\bibitem{fhl18}S.~Filippov, T.~Heinosaari and L.~Lepp\" aj\" arvi, Simulability of observables in general probabilistic theories,
\textit{Phys. Rev.~A.} \textbf{97}, 062102 (2018).
\bibitem{fb94}D.~Foulis and M.~K.~Bennett, Effect algebras and unsharp quantum logics, \textit{Found.~Phys.} \textbf{24}, 133 (1994).
\bibitem{gud73}S.~Gudder,  Convex structures and operational quantum mechanics, \textit{Comm.~Math.~Phys.} \textbf{29}, 249 (1973).
\bibitem{gp98}S.~Gudder and S.~Pulmannov\'a, Representation theorem for convex effect algebras,
\textit{Comment.~Math.~Univ.~Carolinae} \textbf{39}, 659 (1998).
\bibitem{gpbb99}S.~Gudder, S.~Pulmannov\'a, S.~Bugajski and E.~Beltrametti, Convex and linear effect algebras,
\textit{Reports Math.~Phys.} \textbf{44}, 359 (1999).
\bibitem{gbca17}L.~Guerini, J.~Bavaresco, M.~Cunha and A.~Acin, Operational framework for quantum measurement simulability, 
\textit{J.~Math.~Phys.} \textbf{58}, 7092102 (2017).
\bibitem{hz12}T.~Heinosaari and M.~Ziman, \textit{The Mathematical Language of Quantum Theory}, Cambridge University Press, Cambridge 2012.
\bibitem{jp17}A.~Jen\v cov\'a and M.~Pl\'avala, Conditions on the existence of maximally incompatible two-outcome measurements in general probabilistic theory, \textit{Phys. Rev.~A.} \textbf{96}, 022113 (2017).
\bibitem{nam57}I.~Namioka, \textit{Partially Ordered Linear Topological Spaces}, Memoirs, Amer.~Math.~Soc. \textbf{24}, Providence, Rhode Island 1957.
\bibitem{nc00}M.~Nielsen and I.~Chuang, \textit{Quantum Computation and Quantum Information}, Cambridge University Press, Cambridge 2000.
\bibitem{rock70}R.~Rockafellar, \textit{Convex Analysis}, Princeton University Press, Princeton 1970.
\bibitem{sb14}N.~Stevens and P.~Busch, Steering incompatibility, and Bell inequality violations in a class of probabilistic theories, 
\textit{Phys. Rev.~A.} \textbf{89}, 022123 (2014).

\end{thebibliography}
\end{document}